\documentclass[11pt]{article}

\usepackage{amsmath}
\usepackage{amssymb}
\usepackage{amsthm}
\usepackage{graphicx}
\usepackage{authblk}
\usepackage{a4wide}
\usepackage[noadjust]{cite}
\usepackage{caption}
\usepackage{subcaption}
\usepackage[usenames,dvipsnames]{color}

\usepackage[normalem]{ulem}
\definecolor{DarkGreen}{rgb}{0,0.5,0.1} 

\newcommand\soutD{\bgroup\markoverwith
{\textcolor{DarkGreen}{\rule[.5ex]{2pt}{1pt}}}\ULon}
\newcommand\soutP{\bgroup\markoverwith
{\textcolor{blue}{\rule[.5ex]{2pt}{1pt}}}\ULon}

\newcommand{\Hm}[1]{\leavevmode{\marginpar{\tiny
$\hbox to 0mm{\hspace*{-0.5mm}$\leftarrow$\hss}
\vcenter{\vrule depth 0.1mm height 0.1mm width \the\marginparwidth}
\hbox to
0mm{\hss$\rightarrow$\hspace*{-0.5mm}}$\\\relax\raggedright #1}}}

\renewcommand{\d}{\mathrm{d}}
\newcommand{\p}{\partial}
\newtheorem{remark}{Remark}
\newtheorem{conjecture}{Conjecture}
\newtheorem{theorem}{Theorem}
\newtheorem{proposition}{Proposition}
\newtheorem{corollary}{Corollary}
\newtheorem{define}{Definition}
\newtheorem{lemma}{Lemma}
\title{The asymptotic behaviour of the heat equation in a sheared unbounded strip}
\author{Michal Tich\'{y}}
\affil{Department of Mathematics, Faculty of Nuclear Sciences and Physical Engineering, Czech Technical University in Prague}

\begin{document}
\maketitle
\begin{abstract}
 We show that the geometric deformation of shearing yields an improved decay rate for the heat semigroup associated with the Dirichlet Laplacian in an unbounded strip. The proof is based on the Hardy inequality due to the shearing established in \cite{ShearBriet} and the method of self-similar variables and weighted Sobolev spaces for the heat equation.   
\end{abstract}

\section{Introduction}
This paper is motivated by the following general conjecture of
 Krej\v{c}i\v{r}\'{i}k and Zuazua:
\begin{conjecture}[\cite{TwistedTubesZuazua}, Krej\v{c}i\v{r}\'{i}k and Zuazua]
\label{mainConjecture}
Let $\Omega$ be an open connected subset  of $\mathbb{R}^d.$ Let $H_0$ and $H_+$ be two self-adjoint operators in $L^2(\Omega)$ such that $\inf \sigma (H_0) = \inf \sigma (H_+) =0.$ Assume that there is a positive function $\rho : \Omega \to \mathbb{R}$ such that $H_+ \geq \rho,$ while $H_0 - V$ is a negative operator for any non-negative non-trivial $V \in C_0^\infty (\Omega).$ Then there exists a positive function $K: \Omega \to \mathbb{R}$ such that
\begin{equation*}
    \lim_{t \to \infty} \frac{\|e^{-t H_+}\|_{L^2(\Omega, K) \to L^2(\Omega)}}{\|e^{-t H_0}\|_{L^2(\Omega, K) \to L^2(\Omega)}} = 0.
\end{equation*}
\end{conjecture}
In the same paper \cite{TwistedTubesZuazua} the conjecture was proved for a special geometry by showing that a twist of a three-dimensional tube of uniform cross-section yields an improved decay rate for the heat semigroup associated with the Dirichlet Laplacian in the tube with respect to the straight tube. The pioneering work \cite{TwistedTubesZuazua} was followed by a series of papers establishing the validity of the conjecture in various geometric settings 
\cite{DirichletNeumanZuazua}, \cite{withKolb}, \cite{KrejcirikWedges} as well as magnetic environments \cite{KrejcirikMagnetic}, \cite{MagneticCazacu}. 
An alternative version of the conjecture involving
point-wise improved decay rate 
for heat kernels
was stated in \cite{withPinchover} by Fraas, Krej\v{c}i\v{r}\'{i}k and Pinchover. 
The point-wise improvement for twisted tubes was then established by  Grillo, Kova\v{r}\'{i}k and Pinchover in \cite{GrilloKovarik}. In summary, it is expected that the existence of a Hardy inequality for elliptic operators always implies a better decay rate for the generated heat semigroup.
    
It has been shown recently in \cite{ShearBriet} that a Hardy inequality holds 
for the Dirichlet Laplacian in repulsively sheared unbounded strips. In this paper we use this newly established functional inequality and show that Conjecture \ref{mainConjecture} holds in the case of a locally sheared strip.

  The model of \cite{ShearBriet} is characterised by a positive number $d,$ which determines the width of the strip, and a continuous function $f: \mathbb{R} \to \mathbb{R},$ which determines the boundary profile of the strip. The sheared strip $\Omega _f \subset \mathbb{R}^2$ is defined as
\begin{equation*}
    \Omega _f := \{(x,z) \in \mathbb{R}^2 \: | \: f(x) < z < f(x) + d\}.
\end{equation*}
So the strip is built by translating a segment oriented in a constant direction along an unbounded curve in the plane. We see that the boundary of $\Omega _f$ is formed by the curves $x \mapsto (x,f(x))$ and $x \mapsto (x,f(x)+d),$ see Figure \ref{fig:shear}.
\begin{figure}
\centering
\begin{subfigure}{.5\textwidth}
  \centering
  \includegraphics[width=.9\linewidth]{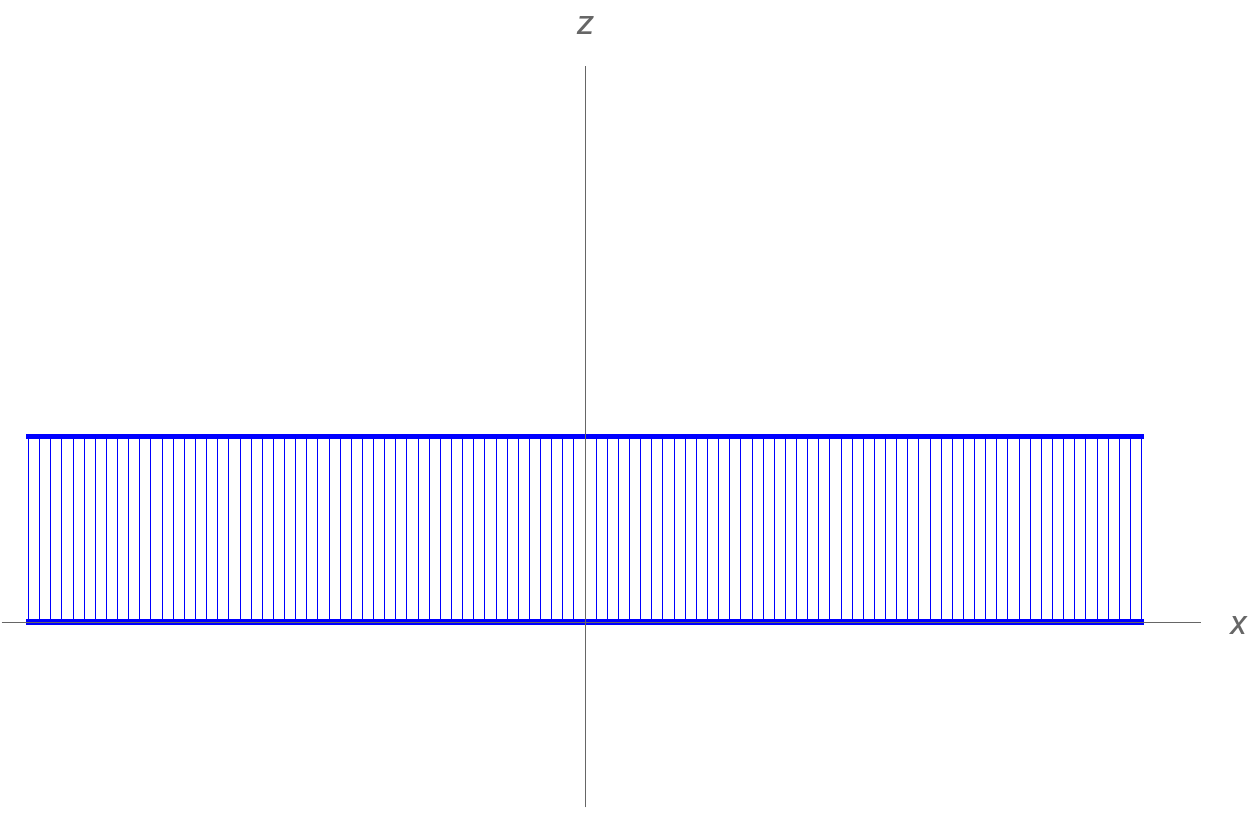}
\end{subfigure}%
\begin{subfigure}{.5\textwidth}
  \centering
  \includegraphics[width=.9\linewidth]{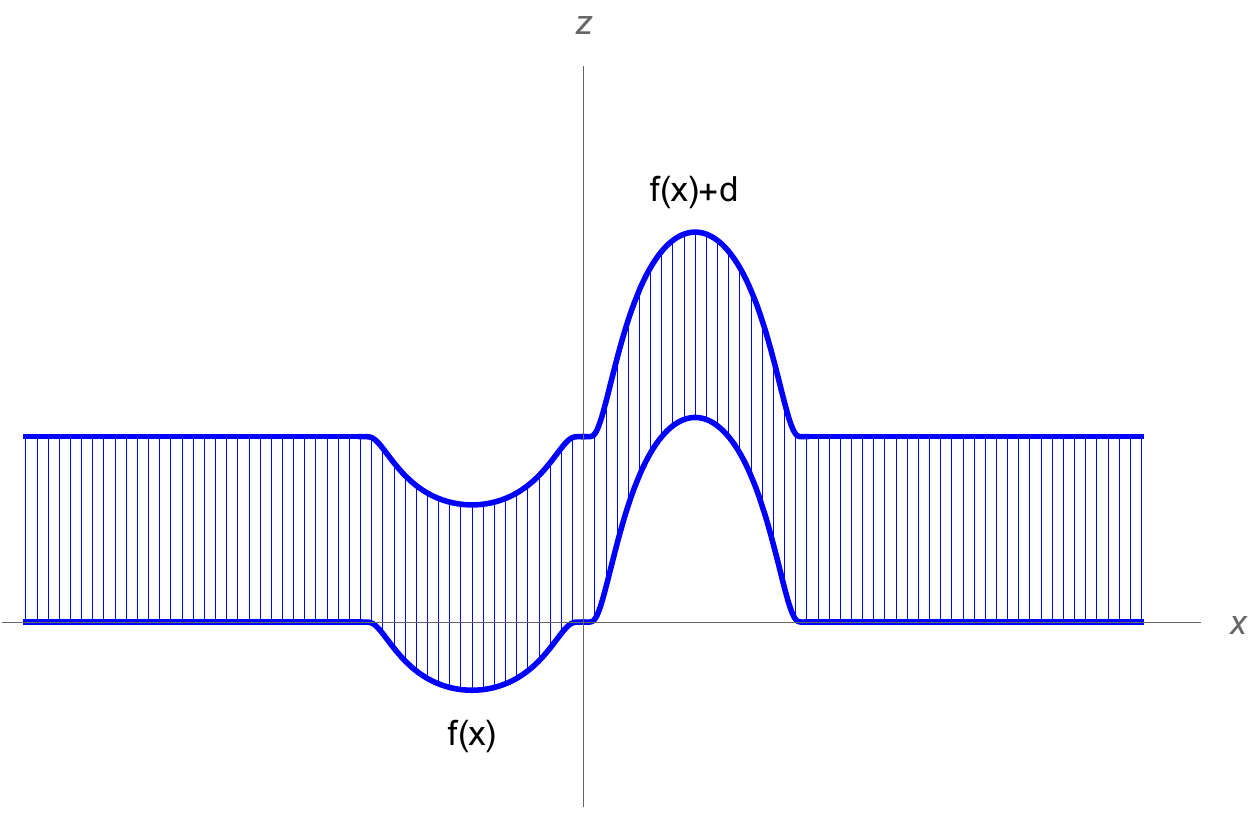}
\end{subfigure}
\caption{The geometry of a sheared strip: a straight strip (left)
and a non-trivially sheared strip (right).}
\label{fig:shear}
\end{figure}
Throughout the paper we assume that $f \in C^{0,1}(\mathbb{R})$ and $f'$ has a compact support.

We consider the heat equation in the sheared strip $\Omega_f$
 \begin{equation}
 \label{heatEquation}
     u_t - \Delta u = 0,
 \end{equation}
 subject to Dirichlet boundary conditions on $\p \Omega_f$ and to the initial condition
 \begin{equation}
 \label{initialConditionU0}
     u(\cdot , 0) = u_0 \in L^2(\Omega_f).
 \end{equation}
Our goal is to show that the solutions of \eqref{heatEquation} converge to the stable equilibrium faster in any non-trivially sheared strip (i.e.\ $f'\not=0$) than in the straight one (corresponding to $f'=0$ identically). The solution to \eqref{heatEquation}--\eqref{initialConditionU0}
 is given by 
 \[u(t) = e^{t \Delta }u_0,\] 
 where $e^{t \Delta }$ is the semigroup operator on $L^2(\Omega_f)$ associated with the Dirichlet Laplacian $-\Delta$ (cf. \cite{SpectralTheoryDavies}, Theorem 5.2.1).  
Under our assumptions, 
it follows from the results established in~\cite{ShearBriet} that
$\sigma(-\Delta)=[E_1,\infty)$,
where $E_1 := \left(\frac{\pi}{d}\right)^2$ is the first eigenvalue of Dirichlet Laplacian on $(0,d)$.
Consequently,
we have for all $t \geq 0$
 \begin{equation*}
     \|e^{t \Delta }\|_{L^2(\Omega_f) \to L^2(\Omega_f)} = e^{-E_1 t}.
 \end{equation*}
 Thus we have the exponential decay rate
 \begin{equation*}
     \|u(t)\|_{L^2(\Omega_f)} \leq e^{-E_1 t}\|u_0\|_{L^2(\Omega_f)},
 \end{equation*}
for each $t \geq 0$ and any initial datum $u_0 \in L^2(\Omega_f).$ Since we are interested in additional time decay properties of the heat semigroup it is natural to consider 
the \textit{shifted} semigroup
\begin{equation*}
    S(t):= e^{t\left(\Delta +E_1\right)},
\end{equation*}
 as an operator from a subspace of $L^2(\Omega_f)$ to $L^2(\Omega_f)$. In this paper we consider the subspace of initial data given by the weighted space
 \begin{equation}
 \label{L2K}
     L^2(\Omega_f,K), \:\text{ where } K(x):= e^{x^2/4}.
 \end{equation}
 As a measure of the additional decay, we consider the polynomial decay rate
\begin{equation*}
    \Gamma (\Omega_f):= \sup \{\Gamma \:| \:\exists \: C_\Gamma >0, \forall t \geq 0, \:  \|S(t)\|_{L^2(\Omega_f,K) \to L^2(\Omega_f)} \leq C_\Gamma (1+t)^{-\Gamma}\}.
\end{equation*}
The main result in this paper reads as follows:
\begin{theorem}
\label{mainResult}
Let  $f \in C^{0,1}(\mathbb{R})$ and $f'$ has compact support. We have
\begin{equation}
   \Gamma(\Omega_f)\begin{cases}
   = 1/4 \quad \mathrm{ if } \: f' =0 \quad \quad \:\mathrm{ (straight \: strip),}\\
   \geq 3/4 \quad \mathrm{ if } \: f' \neq 0 \quad \quad \:\mathrm{ (sheared \: strip).}
   \end{cases}
\end{equation}
\end{theorem}
Theorem \ref{mainResult} can be reformulated as follows. For every $\Gamma <\Gamma(\Omega_f)$ there exists $ C_\Gamma >0$ such that
\begin{equation}
    \label{mainResultRefolmulated}
    \|u\|_{L^2(\Omega_f)} \leq C_\Gamma (1+t)^{-\Gamma}e^{-E_1t} \|u_0\|_{L^2(\Omega_f,K)},
\end{equation}
for each $t \geq 0$ and any initial datum $u_0 \in L^2(\Omega_f,K).$

The decay rate $1/4$ for straight strips was already proved for example in \cite[\mbox{Proposition~1}]{DirichletNeumanZuazua}. It is the at least three times better improvement for sheared strips which is new here. The proof of Theorem  \ref{mainResult} is based on the method of self-similar solutions developed in the
whole Euclidean space by Escobedo and Kavian \cite{EscobedoKavian} and used to prove Conjecture~\ref{mainConjecture} in many special cases (cf. \cite{TwistedTubesZuazua}, \cite{DirichletNeumanZuazua}, \cite{KrejcirikMagnetic}, \cite{withKolb}, \cite{KrejcirikWedges} and \cite{MagneticCazacu} ).
Using the self-similar transformation we reformulate \eqref{heatEquation} in the weighted space \eqref{L2K} and show that the associated generator has a compact resolvent. Finally we look at the asymptotic behaviour of \eqref{heatEquation} as the self-similar time tends to infinity. The crucial ingredient in the proof is the existence of a Hardy inequality due to \cite{ShearBriet} for the Dirichlet Laplacian in our setting.

The organisation of this paper is as follows. In the following Section 2 we give a precise definition
of the Dirichlet Laplacian in the sheared strip and transform it into the straight strip via curvilinear coordinates. Furthermore, we state the Hardy inequality for locally sheared strip. The main body of the paper is represented by Section 3 where we develop the method of self-similar solutions to get the improved decay rate of Theorem~\ref{mainResult}. Moreover, we  establish an alternative result in Theorem~\ref{alternativeResult}. 
The paper is concluded in Section~\ref{Conclusion} 
by commenting on physical motivations and mentioning some open problems.

\section{Preliminaries}
We consider the Dirichlet Laplacian $-\Delta$  which is introduced standardly as the self-adjoint operator in the Hilbert space $L^2(\Omega_f)$ associated with the quadratic form 
\begin{equation*}
    \mathcal{Q}^{\Omega_f}_D[u]:= \int_{\Omega_f}|\nabla u|^2, \quad \mathcal{D}(\mathcal{Q}^{\Omega_f}_D ):= H^1_0(\Omega_f),
\end{equation*}
using the fact  that $\Omega_f$ is an open set, which will be seen in this section beneath.
We would like to express the Dirichlet Laplacian in the sheared geometry $\Omega _f$ in natural curvilinear coordinates. By denoting $\Omega:= \mathbb{R}\times (0,d)$ we identify $\Omega _f$ with $\mathcal{L}(\Omega),$ where $\mathcal{L}: \mathbb{R}^2 \to \mathbb{R}^2$ is the shear mapping defined by
\begin{equation}
\label{TrafoL}
    \mathcal{L}(x,z):= (x,f(x)+z).
\end{equation}
The corresponding metric has form 
\begin{equation*}
    g:= \nabla \mathcal{L}\cdot (\nabla \mathcal{L})^T = \begin{pmatrix}
    1+f'^2& f' \\
    f'& 1
    \end{pmatrix},
\end{equation*}
where $\cdot$ denotes the matrix product.
It is easy to see that $\det (g)= 1.$
Due to the assumption $f \in C^{0,1}(\mathbb{R}),$ the shear mapping $\mathcal{L}:\Omega \to \Omega_f$ is a local diffeomorphism (cf. \cite[Page 5]{ShearBriet}). Because of the injectivity of $\mathcal{L}$ it is a global diffeomorphism. Therefore, $\Omega_f$ is an open set. So we can identify $\Omega_f$ with the Riemann manifold $(\Omega, g).$

Next, we can define the unitary transformation 
\begin{equation*}
    \mathcal{U}: L^2(\Omega_f) \to L^2(\Omega): \{u \mapsto u \circ \mathcal{L}\}.
\end{equation*}
Thus the Dirichlet Laplacian $-\Delta$ is unitary equivalent and therefore isospectral to the operator
\begin{equation*}
H:= \mathcal{U}\left(-\Delta\right) \mathcal{U}^{-1}   
\end{equation*}
in the Hilbert space $L^2(\Omega).$ This operator is associated with the quadratic form 
\begin{equation*}
    \mathcal{Q}^H[v]:= \mathcal{Q}^{\Omega_f}_D[\mathcal{U}^{-1}v], \quad \mathcal{D}( \mathcal{Q}^H) := \mathcal{U}\:\mathcal{D}(\mathcal{Q}^{\Omega_f}_D).
\end{equation*}
Finally, we overtake from \cite{ShearBriet} the following proposition and state the Hardy-type inequality for the sheared strip in our setting :
\begin{proposition}[\cite{ShearBriet}, Proposition 1]
Let $u \in C_0^\infty(\Omega_f),$ a core of $\mathcal{Q}^{\Omega_f}_D.$ Then $v := u \circ \mathcal{L}$ is compactly supported, $v \in H^1_0(\Omega)$ and 
\begin{equation*}
   \mathcal{Q}^H[v] = \|\p_xv-f'\p_z v\|_{L^2(\Omega)}^2 + \|\p_z v\|_{L^2(\Omega)}^2,  
\end{equation*}
where $f'$ represents the function $f'\otimes1$ on $\mathbb{R}\times (0,d).$
Furthermore, $C_0^\infty(\Omega)$ is a core of $ \mathcal{Q}^H$ and for $f' \in L^\infty(\mathbb{R})$ we have $\mathcal{D}(  \mathcal{Q}^H) = H^1_0(\Omega).$
\end{proposition}
In a distributional sense we have
\begin{equation}
\label{LaplaceCurvilinear}
    H = -(\p_x -f'\p_z )^2-\p_z^2.
\end{equation}
\begin{theorem}[\cite{ShearBriet}, Theorem 4]
\label{TheoremHardy}
Let $f \in C^{0,1}(\mathbb{R})$ and has compact support. If $f' \neq 0,$ then there exists a positive constant $c_H$ such that the inequality
\begin{equation}
\label{Hardy}
    -\Delta  - E_1 \geq \frac{c_H}{1+x^2}
\end{equation}
holds in the sense of quadratic forms in $L^2(\Omega_f).$
\end{theorem}
Denoting $\rho(x):= \frac{1}{\sqrt{1+x^2}},$ the Hardy inequality \eqref{Hardy} means
\begin{equation}
\label{RemarkHardy}
    \|\nabla \psi\|^2_{L^2(\Omega_f)}-E_1\|\psi\|^2_{L^2(\Omega_f)} \geq c_H\|\rho \psi\|^2_{L^2(\Omega_f)},
\end{equation}
for all $\psi \in H_0^1(\Omega_f).$

\section{Self-similarity transformation}
Recalling \eqref{LaplaceCurvilinear} the shifted heat equation in the curvilinear coordinates has form
\begin{equation}
\label{evolutionOriginal}
    u_t+H u-E_1 u = 0 \quad \text{in } \Omega \times (0,\infty),
\end{equation}
subject to the Dirichlet boundary condition on $\p \Omega$ and to the initial condition $u(0) = u_0 \in L^2(\Omega).$ 
The weak formulation has form
\begin{equation}
\label{weakformulationOriginal}
\begin{aligned}
   &\langle v, \dot{u}(t) \rangle + \left((\p _x - f'(x)\p _z)v,(\p _x - f'(x)\p _z)u(t)\right)_{L^2(\Omega)}\\
   &+(\p_z v, \p_z u(t))_{L^2(\Omega)} -E_1 ( v,  u(t))_{L^2(\Omega)} = 0,   
\end{aligned}
\end{equation}
for every $v \in H^1_0(\Omega)$ and a.e. $t \in [0,\infty),$ with $u(0) = u_0 \in L^2(\Omega),$ where $\langle \cdot , \cdot \rangle$ denotes the pairing of $H_0^1(\Omega) $ and $H^{-1}(\Omega).$ We know that the solution $u$ belongs to $C^0([0,\infty),L^2(\Omega))$ by the semigroup theory (cf. \cite[Corollary~ 2.3]{PazySemigroups}).

\subsection{Change of variables}
We perform the so called self-similarity transformation, developed in the whole Euclidean space by Escobedo and Kavian \cite{EscobedoKavian}. This approach was also used in the papers \cite{TwistedTubesZuazua}, \cite{DirichletNeumanZuazua}, \cite{KrejcirikMagnetic}, \cite{withKolb}, \cite{KrejcirikWedges} and \cite{MagneticCazacu}. Following \cite{TwistedTubesZuazua} we perform the self-similarity transformation in the first space variable only: 
\begin{align}
\label{selfSimilarTrafo}
    w(y,z,s) &= s^{s/4}u(s^{s/2}y,z,e^{s}-1),\\
    y &= x e^{-s/2},\\
    s &= \ln(1+t).
\end{align}
Consequently, in our case the self-similar transformation is a  unitary transformation $U$ on $L^2(\Omega)$ which maps every solution $u \in L^2_{\text{loc}}\left((0,\infty), \d t; L^2(\Omega, \d x \d z)\right)$ to a solution $w:= Uu$ in a new $s$-time weighted space $L^2_{\text{loc}}\left((0,\infty), e^s\d s; L^2(\Omega, \d y \d z)\right).$
It is easy to check that, in the new variables, the evolution \eqref{evolutionOriginal}
 is described by 
\begin{equation}
\label{evolutionSelfsimilar}
    \p_s w - \frac{1}{2}y\p_y w - (\p_y-\sigma_s\p_z)^2w - e^s\p^2_{z}w - E_1 e^s w - \frac{1}{4}w = 0,
\end{equation}
where $\sigma_s(y) := e^{s/2}f'(e^{s/2}y)$ and $(y,z)$ play the role of space variables and $s$ plays the role of a new time.
More precisely the weak formulation \eqref{weakformulationOriginal} is transferred into
\begin{equation}
\label{weakformulationSeflsimilar}
    \bigg\langle v, \dot{w}(s)-\frac{1}{2}y\p_yw(s) \bigg\rangle + \mathcal{Q}_s(v,w(s))-E_1 e^s(v,w(s))_{L^2(\Omega)}=0,
\end{equation}
for every $v \in H_0^1(\Omega)$ and a.e. $s \in [0,\infty),$ with $w(0) = w_0:= U u_0 = u_0,$ where the sesquilinear form $\mathcal{Q}_s(\cdot, \cdot)$ is associated with 
\begin{equation}
    \mathcal{Q}_s[w]:= \|\p_yw-\sigma_s\p_z w\|^2_{L^2(\Omega)}+e^s\|\p_z w\|^2_{L^2(\Omega)}- \frac{1}{4}\|w\|^2_{L^2(\Omega)},
\end{equation}
for every $w \in \mathcal{D}(\mathcal{Q}_s):= H^1_0(\Omega).$
\subsection{The evolution in the natural weighted space}
Since the unitary transformation $U$ on $L^2(\Omega)$ preserves the space norm of solutions \eqref{weakformulationOriginal} and \eqref{weakformulationSeflsimilar},  we can analyze the asymptotic time behaviour in the new variables. However, the natural space for studying the evolution described by \eqref{weakformulationSeflsimilar} is the weighted space $L^2(\Omega, e^{y^2/4}\d y \d z)$ instead of $L^2(\Omega)$.
\begin{define}
For $k \in \mathbb{Z}$ we define the weighted space
\begin{equation}
    \mathcal{H}_k:= L^2(\Omega, K^k(y)\d y \d z),
\end{equation}
where $K(y) = e^{y^2/4}.$
In the similar way we define the weighted Sobolev space
\begin{equation}
 \mathcal{H}^1_k:= H^1_0(\Omega, K^k(y)\d y \d z),   
\end{equation}
as a closure of $C^\infty_0(\Omega)$ with respect to the norm $(\|\cdot\|^2_{\mathcal{H}_k} + \|\nabla \cdot\|^2_{\mathcal{H}_k})^{1/2}.$ Finally, we define $\mathcal{H}_k^{-1}$ as a dual space of $\mathcal{H}_k^{1}.$
\end{define}
As a  next step, we want to reconsider the evolution \eqref{evolutionSelfsimilar} as a problem posed in the weighted space $\mathcal{H}_1$ instead of $\mathcal{H}_0 = L^2(\Omega).$ We start with a formal calculation. In the equation \eqref{weakformulationSeflsimilar} we choose $\tilde{v}(y,z) = e^{y^2/4}v(y,z)$ as a test function, where $v \in C_0^\infty(\Omega)$ is arbitrary. The sesquilinear form $\mathcal{Q}_s(\tilde{v},w(s))$ reads
\begin{align}
\label{formalComputing}
    \mathcal{Q}_s(\tilde{v},w(s)) &= (\p_y v-\sigma_s\p_z v,\p_y w(s)-\sigma_s\p_z w(s))_{\mathcal{H}_1}+\bigg( \frac{1}{2}yv,\p_y w(s)-\sigma_s\p_z w(s)\bigg)_{\mathcal{H}_1}  \\
    &+ e^s(\p_z v,\p_z w(s))_{\mathcal{H}_1} -\frac{1}{4}(v,w(s))_{\mathcal{H}_1},
\end{align}
which means that the weak formulation of evolution \eqref{weakformulationSeflsimilar} in the weighted space $\mathcal{H}_1$  has the form
\begin{equation}
\label{weakformulationH1}
    \langle v, \dot{w}(s) \rangle + a_s(v,w(s)) = 0,
\end{equation}
where $\langle \cdot, \cdot \rangle$ denotes the pairing of $\mathcal{H}^{1}_1$ and $\mathcal{H}^{-1}_1$ and
\begin{align*}
      a_s(v,w) &=  (\p_y v-\sigma_s\p_z v,\p_y w-\sigma_s\p_z w)_{\mathcal{H}_1}-\frac{1}{2}( yv,\sigma_s\p_z w)_{\mathcal{H}_1} \\&+ e^s(\p_z v, \p _z w)_{\mathcal{H}_1} -E_1e^s( v, w)_{\mathcal{H}_1} -\frac{1}{4}(v,w)_{\mathcal{H}_1}.
\end{align*}
We see that the form $a_s$ is not symmetric. Next, we show that the problem \eqref{weakformulationH1} is well posed in $\mathcal{H}_1$ and also the solution solves the transformed original problem \eqref{weakformulationSeflsimilar}:
\begin{proposition}
\label{propositionUniq}
Let $w_0 \in \mathcal{H}_1$ be an arbitrary function. Then there exists an unique function $w$ such that:
\begin{enumerate}
    \item $w \in L^2_{\text{loc}}\left((0, \infty); \mathcal{H}^1_1\right)\cap C^0([0,\infty),\mathcal{H}_1),$
    \item $\dot{w} \in L^2_{\text{loc}}\left((0, \infty); \mathcal{H}^{-1}_1\right), $
    \item $w$ satisfies 
    $\langle v, \dot{w}(s) \rangle + a_s(v,w(s)) = 0$
    for all $v \in \mathcal{H}^1_1$ and a.e.  $s \in [0, \infty),$
    \item $w(0) = w_0.$
\end{enumerate}
\end{proposition}
In the proof we use the following theorem:
\begin{theorem}[\cite{Lions}, Chapter 3, Theorem 4.1]
\label{theoremLions}
Let $\mathcal{H}$ be a Hilbert space with the scalar product $(\cdot , \cdot)_\mathcal{H}$ and the norm $\|\cdot\|_\mathcal{H},$ where the dual space $\mathcal{H}^*$ is identified with $\mathcal{H}$. Let $\mathcal{V}$ be a separable Hilbert space with the norm $\|\cdot\|_\mathcal{V}$ and let $\mathcal{V} \subset  \mathcal{H}$ with dense and continuous injection, thus $\mathcal{V} \subset \mathcal{H} \subset \mathcal{V}^*.$ We have a continuous sesquilinear form $a_t(u,v): \mathcal{V} \times \mathcal{V} \to \mathbb{C}$ for a.e. $t \in [0,T],$ where $T>0$ is fixed, which satisfies the following properties:
\begin{enumerate}
    \item $\forall u, v \in \mathcal{V},$ the function $t \mapsto a_t(u,v)$ is measurable,
    \item $|a_t(u,v)| \leq C\|u\|_\mathcal{V}\|v\|_\mathcal{V}$ for a.e. $t \in [0,T]$ and $\forall u, v \in \mathcal{V},$
    \item $\mathfrak{R}\{a_t(u,u)\} \geq c_1\|u\|^2_\mathcal{V} - c_2\|u\|^2_\mathcal{H}$ for a.e. $t \in [0,T]$ and $\forall u \in \mathcal{V},$
\end{enumerate}
where $C,c_1, c_2$ are constants and $c_1 >0.$ 
Then for given $f \in L^2((0,T);\mathcal{V}^*)$ and $u_0 \in \mathcal{H}$ there exists a unique function $u$ satisfying 
\begin{enumerate}
    \item $u \in L^2((0,T);\mathcal{V})\cap C([0,T];\mathcal{H}),$
    \item $\dot{u} \in L^2((0,T);\mathcal{V}^*),$
    \item $\langle v,\dot{u}(t)\rangle + a_t(v,u(t)) = \langle v, f(t) \rangle$ for a.e. $t \in (0,T),$ and $\forall v \in \mathcal{V},$
    \item $u(0) = u_0.$
\end{enumerate}
\end{theorem}
\begin{proof}[Proof of Proposition \ref{propositionUniq}]
The proof is inspired by the proof of Proposition 5.1 from the paper \cite{TwistedTubesZuazua}. First, we show that the sesquilinear form $a_s$ is well defined with the domain $\mathcal{D}(a_s):= \mathcal{H}^1_1$ for any fixed $s \in [0,\infty)$ and thus it is continuous. Using the fact that $\sigma_s$ is bounded  for every finite $s$ we only have to show that for every $ v \in \mathcal{H}^1_1$ we have $ y v \in \mathcal{H}_1.$ For $v \in C^\infty_0(\Omega)$ we obtain
\begin{align*}
    \|y v\|_{\mathcal{H}_1}^2 &= 2\int_\Omega y|v(y, z)|^2\frac{\d (e^{y^2/4})}{\d y} \d y \d z   \\
    &=  -2\int_\Omega \left(|v|^2 + 2y\mathfrak{R}[\bar{v}\p_y v]\right)e^{y^2/4}\d y \d z  \\
    &\leq 4\bigg|\int_\Omega y\mathfrak{R}[\bar{v}\p_y v] e^{y^2/4}\d y \d z \bigg|  \\
    &\leq 4(|yv|, |\p_yv|)_{\mathcal{H}_1} \\
    &\leq 4 \|yv\|_{\mathcal{H}_1} \|\p_yv\|_{\mathcal{H}_1},
\end{align*}
where we used the Cauchy-Schwarz inequality in the last estimate. Consequently,
\begin{equation}
\label{estimateYV}
   \|yv\|_{\mathcal{H}_1} \leq 4  \|\p_yv\|_{\mathcal{H}_1} \leq  4  \|v\|_{\mathcal{H}^1_1}.
\end{equation}
Next, the inequality can be extended to all $v \in \mathcal{H}^1_1$ by density argument. Therefore, the sesquilinear form $a_s$ is well defined for every $ s \geq 0$ and $ \forall v,w \in \mathcal{H}^1_1.$ To prove the rest of the proposition we use Theorem \ref{theoremLions} which was mentioned above. In our case $\mathcal{V} = \mathcal{H}^1_1$ and $\mathcal{H} = \mathcal{H}_1.$ We show that $a_s(\cdot , \cdot)$ satisfies all three assumptions of Theorem \ref{theoremLions}.
First, it is easy to see that the function $s\mapsto a_s(v,w)$ is continuous on $[0,\infty),$ $\forall v,w \in \mathcal{H}^1_1,$ therefore, it is also measurable on $[0,\infty).$
    Furthermore, for $s_0 \in [0,\infty)$ fixed we get:
    \begin{align*}
        &|(\p_yv-\sigma_{s_0} \p_z v, \p_yw-\sigma_{s_0} \p_z w )_{\mathcal{H}_1}| \leq (1+\|\sigma_{s_0}\|_{L^\infty(\Omega)})\|v\|_{\mathcal{H}^1_1}\|w\|_{\mathcal{H}^1_1},
    \end{align*}
    where we used the Cauchy-Schwarz inequality and the boundedness of $f'.$ Next, we again use the Cauchy-Schwarz inequality and also the inequality \eqref{estimateYV} from above:
    \begin{align*}
        |(yv,\sigma_{s_0}\p_z w)_{\mathcal{H}_1}| &\leq 4 \|\sigma_{s_0}\|_{L^\infty(\mathbb{R})}\|  v\|_{\mathcal{H}^1_1}\|w\|_{\mathcal{H}^1_1},\\
        |e^{s_0}(\p_z v, \p_z w)_{\mathcal{H}_1}| &\leq e^{s_0}\|  v\|_{\mathcal{H}^1_1}\| w\|_{\mathcal{H}^1_1}, \\
        |( v, w)_{\mathcal{H}_1}| &\leq \|  v\|_{\mathcal{H}^1_1}\| w\|_{\mathcal{H}^1_1}.
    \end{align*}
     Summing up, we have shown that
        \begin{equation}
        \label{boundenessLions}
         |a_s(v,w)| \leq C \|v\|_{\mathcal{H}^1_1}\| w\|_{\mathcal{H}^1_1},
        \end{equation}
        for every $ s \in [0,s_0]$ and for all $ v,w \in \mathcal{H}^1_1, $ where $C$  is a constant depending on $s_0, \|f'\|_{L^\infty(\Omega)}$ and $E_1.$
        Finally, we have to show that
        \begin{equation}
        \label{coercivityLions}
            \mathfrak{R}\{a_s[v]\} \geq c_1\|v\|^2_{\mathcal{H}_1^1} - c_2\|v\|^2_{\mathcal{H}_1}
        \end{equation}
        for every $v \in \mathcal{H}^1_1$ and a.e. $s \in [0,s_0],$ where $a_s[v]:= a_s(v,v).$ We have
        \begin{equation}
        \label{realForm_as}
           \mathfrak{R}\{a_s[v]\} =  \|\p_y v-\sigma_s\p_z v\|^2_{\mathcal{H}_1} + e^s\|\p_z v\|^2_{\mathcal{H}_1} -E_1e^s\|v\|^2_{\mathcal{H}_1} -\frac{1}{4}\|v\|^2_{\mathcal{H}_1} -\frac{1}{2}\mathfrak{R}( yv,\sigma_s\p_z w)_{\mathcal{H}_1}.
        \end{equation}
        For $v \in C^\infty_0(\Omega)$ an integration by parts shows that
        \begin{align}
        \label{nonsymmetricpartZero}
            \mathfrak{R}( yv,\sigma_s\p_z v)_{\mathcal{H}_1} = 0
        \end{align}
        and by density this result can be extended to all $ v \in \mathcal{H}^1_1.$ As a next step, we would like to estimate the term $ \|\p_y v-\sigma_s\p_z v\|^2_{\mathcal{H}_1}$ using the trivial inequality $(a-b)^2 \geq \epsilon a^2 - \frac{\epsilon}{1-\epsilon}b^2:$
        \begin{align*}
            \|\p_y v-\sigma_s\p_z v\|^2_{\mathcal{H}_1} &\geq \epsilon\|\p_y v\|^2_{\mathcal{H}_1} - \frac{\epsilon}{1-\epsilon}\|\sigma_s\p_z v\|^2_{\mathcal{H}_1}  \\
            &\geq \epsilon\|\p_y v\|^2_{\mathcal{H}_1} - \frac{\epsilon}{1-\epsilon}e^s\|f'\|_{L^\infty(\mathbb{R})}\|\p_z v\|^2_{\mathcal{H}_1}.
        \end{align*}
        In the next step we use the Poincar\'e inequality on $(0,d):$
        \begin{equation}
            \label{poincare}
            \|\nabla g\|^2_{L^2((0,d))} \geq E_1 \| g\|^2_{L^2((0,d))}, \quad \forall g \in H_0^1((0,d)),
        \end{equation}
        and the Fubini's theorem
        \begin{align*}
         \|\p_y v-\sigma_s\p_z v\|^2_{\mathcal{H}_1} +(1-\epsilon)e^s\|\p_z v\|^2_{\mathcal{H}_1} &\geq  \epsilon\|\p_y v\|^2_{\mathcal{H}_1} +e^s\left( (1-\epsilon)- \frac{\epsilon}{1-\epsilon}\|f'\|_{L^\infty(\mathbb{R})}\right)\|\p_z v\|^2_{\mathcal{H}_1} \\
         &\geq \epsilon\|\p_y v\|^2_{\mathcal{H}_1} + E_1e^s\left(1-\epsilon - \frac{\epsilon}{1-\epsilon}\|f'\|_{L^\infty(\mathbb{R})}\right)\| v\|^2_{\mathcal{H}_1},
        \end{align*}
        where the last inequality holds for sufficiently small $\epsilon$ such that 
        \[(1-\epsilon)- \frac{\epsilon}{1-\epsilon}\|f'\|_{L^\infty(\mathbb{R})} > 0.\]
        Next, using this inequality, the identity \eqref{nonsymmetricpartZero} and the trivial bound $1\leq e^s \leq e^{s_0}, \forall s \in [0,s_0],$ we get the estimate for \eqref{realForm_as},  $\forall v \in \mathcal{H}^1_1$ and sufficiently small $\epsilon:$
        \begin{equation*}
            \mathfrak{R}\{a_s[v]\} \geq \epsilon\|\p_y v\|^2_{\mathcal{H}^1_1} - \left[E_1e^{s_0}\left(\epsilon + \frac{\epsilon}{1-\epsilon}\|f'\|_{L^\infty(\mathbb{R})}\right)+\frac{1}{4}+\epsilon\right]\| v\|^2_{\mathcal{H}_1},
        \end{equation*}
          where the constant 
          \[c_1:= \epsilon\] depends on $\|f'\|_{L^\infty(\mathbb{R})}$ and the constant 
          \[c_2 := \left[E_1e^{s_0}\left(\epsilon + \frac{\epsilon}{1-\epsilon}\|f'\|_{L^\infty(\mathbb{R})}\right)+\frac{1}{4}+\epsilon\right] \] depends on $s_0, \|f'\|_{L^\infty(\mathbb{R})}$ and $E_1.$
          
Using Theorem \ref{theoremLions} we conclude that the unique solution $w$ of \eqref{weakformulationH1} satisfies 
  \begin{align*}
       &w \in L^2((0, s_0);\mathcal{H}^1_1)\cap C^0([0,s_0];\mathcal{H}_1), \\
       &\dot{w} \in L^2((0, s_0);\mathcal{H}^{-1}_1).
  \end{align*}
  Using the fact that $s_0$ is an arbitrary positive number, we obtain the global continuous solution
  \[w \in C^0([0,\infty);\mathcal{H}_1).\]
This concludes the proof of Proposition~\ref{propositionUniq}.
\end{proof}
\begin{remark}
\label{PropositionClosed_as}
Using the estimates \eqref{boundenessLions} and \eqref{coercivityLions} we get that the sesquilinear form $a_s$ is closed on its domain $\mathcal{H}^1_1.$
\end{remark}
Next, we can prove a partial equivalence of evolutions \eqref{weakformulationSeflsimilar} and \eqref{weakformulationH1}.
\begin{proposition}
Let $w_0\in \mathcal{H}_1$ and let $w$ be the unique solution of \eqref{weakformulationH1}, $\forall v \in \mathcal{H}^1_1$ and a.e. $s \in [0, \infty),$ subject to the initial condition $w(0) = w_0,$ that is specified in Proposition \ref{propositionUniq}. Then $w$ is also the unique solution of \eqref{evolutionSelfsimilar}, $\forall \tilde{v} \in \mathcal{H}^1_0$ and a.e. $s \in [0, \infty),$ subject to the same initial condition.
\end{proposition}
\begin{proof}
We choose a test function $v(y,z):= K^{-1}(y)\tilde{v}(y,z) $ in \eqref{weakformulationH1}, where $\tilde{v} \in C^\infty_0(\Omega)$ is an arbitrary function. Recalling the formal computation in \eqref{formalComputing} it is easy to see that $w$ satisfies also the equation \eqref{weakformulationSeflsimilar}, $\forall \tilde{v} \in C^\infty_0(\Omega)$ and a.e. $s \in [0,\infty).$ By the density argument this result can be extended to all $ \tilde{v} \in \mathcal{H}^1_0.$
\end{proof}
\subsection{Reduction to a spectral problem}
As a consequence of the previous subsection, reducing the space of initial data, we can focus on the asymptotic time behaviour of the solutions of \eqref{weakformulationH1}. By choosing $v:=w(s)$ in \eqref{weakformulationH1} and combining the equation with its conjugate version we get
\begin{equation}
\label{IdentityJ1}
    \frac{1}{2}\frac{\d}{\d s}\|w(s)\|^2_{\mathcal{H}_1} = -J^{(1)}_s[w(s)],
\end{equation}
where $J^{(1)}_s[w(s)] = \mathfrak{R}\{a_s[w]\}$ and $w \in \mathcal{D}(J^{(1)}):= \mathcal{D}(a_s) = \mathcal{H}^1_1.$ Recalling \eqref{realForm_as} and \eqref{nonsymmetricpartZero} we obtain
\begin{equation}
\label{formJ1}
 J^{(1)}_s[w]  = \|\p_y w-\sigma_s\p_z w\|^2_{\mathcal{H}_1} + e^s\|\p_z w\|^2_{\mathcal{H}_1} -E_1e^s\|w\|^2_{\mathcal{H}_1} -\frac{1}{4}\|w\|^2_{\mathcal{H}_1}.
\end{equation}
Similarly as in Remark \ref{PropositionClosed_as}, using the estimates \eqref{boundenessLions} and \eqref{coercivityLions} we get that the form $J^{(1)}_s$ is closed on its domain $\mathcal{H}^1_1.$
As a next step, we would like to analyze the coercivity of this form. We use the spectral bound valid for each fixed $ s \in [0,\infty):$
\begin{equation}
\label{spectralboundJ1}
J^{(1)}_s[w] \geq \mu (s)\|w\|^2_{\mathcal{H}_1},   \quad \forall w \in \mathcal{H}^1_1,
\end{equation}
where $\mu(s)$ is the lowest point of the spectrum of the operator $T^{(1)}_s$ in $\mathcal{H}_1$ associated with $J^{(1)}_s$ via the representation theorem (cf. \cite[Chapter 6, Theorem 2.1]{Kato}).
\begin{proposition}
The operator $T^{(1)}_s$ is self-adjoint.
\end{proposition}
\begin{proof}
We know that the form $J^{(1)}_s$ is closed and densely defined, due to the density of $\mathcal{H}^1_1$ in $\mathcal{H}_1.$ Using the estimate \eqref{coercivityLions} we see that the form is also bounded from below. Therefore, the remark follows from the representation theorem (cf. \cite[Chapter 6, Theorem 2.6.]{Kato}). 
\end{proof}
We take the identity \eqref{IdentityJ1},  replace $-J^{(1)}_s[w(s)]$ by the spectral bound \eqref{spectralboundJ1} and integrate:
\begin{equation}
\label{energyEstimate}
   \|w(s)\|^2_{\mathcal{H}_1} \leq \|w_0\|^2_{\mathcal{H}_1}e^{-\int^s_0\mu(\tau) \d \tau}. 
\end{equation}
We see that we reduced the problem of asymptotic time behaviour of \eqref{evolutionSelfsimilar} to a spectral analysis of the family of the operators $\{T^{(1)}_s\}_{s \geq 0}.$

Next, we map the operator $T^{(1)}_s$ in $\mathcal{H}_1$ into the unitary equivalent operator $T^{(0)}_s$ in $\mathcal{H}_0$ via the unitary transformation $\mathcal{U}_0:\mathcal{H}_1 \to \mathcal{H}_0$ defined by:
\begin{equation}
\label{trafoU0}
    (\mathcal{U}_0w)(y,z):= K^{1/2}(y)w(y,z) = e^{y^2/8}w(y,z).
\end{equation}
We define $T^{(0)}_s:=\mathcal{U}_0 T^{(1)}_s\mathcal{U}_0^{-1}, $ which is the self-adjoint operator associated with the quadratic form $J^{(0)}_s[v]:=J^{(1)}_s[\mathcal{U}_0^{-1}v], $ where $v \in \mathcal{D}(J^{(0)}_s):= \mathcal{U}_0\mathcal{D}(J^{(1)}_s).$ A straightforward calculation yields
\begin{equation}
\label{formJ0}
  J^{(0)}_s[v]= \|\p_yv-\sigma_s\p_z v\|^2_{\mathcal{H}_0}  +\frac{1}{16}\|y v\|^2_{\mathcal{H}_0}  +e^s\|\p_z v\|^2_{\mathcal{H}_0}  -E_1e^s\| v\|^2_{\mathcal{H}_0}  
\end{equation}
for all $v \in \mathcal{D}(J^{(0)}_s).$
Moreover, it is easy to verify that the domain of the form $J^{(1)}_s$ is in fact the closure of $C^\infty_0(\Omega)$ with respect to the norm $\|\cdot\|_{J^{(0)}}:= \left(\|\cdot\|^2_{\mathcal{H}_0}+\|\nabla \cdot\|^2_{\mathcal{H}_0}+\|y\cdot\|^2_{\mathcal{H}_0}\right)^{1/2}.$
We see that the domain $\mathcal{D}(J^{(1)}_s)$ is independent of $s.$ Finally, we show that $\mu (s)$ is the lowest eigenvalue of $T_s^{(1)}:$
\begin{proposition}
The operators $T_s^{(1)} \simeq T_s^{(0)}$ have purely discrete spectrum for every $s \in [0, \infty).$
\end{proposition}
\begin{proof}
First, we define the operator $L$ and the corresponding quadratic form 
\begin{equation*}
    \mathcal{Q}^L[v]:= \|\p_y v \|^2_{\mathcal{H}_0} + \|\p_z v \|^2_{\mathcal{H}_0} +\frac{1}{16}\|y v\|^2_{\mathcal{H}_0}
\end{equation*}
for all $v \in \mathcal{D}(\mathcal{Q}^L) := \mathcal{D}(\mathcal{J}_s^{(0)}).$ Using the fact that the harmonic-oscillator Hamiltonian 
\begin{equation}
\label{harmonicOscillator1D}
    l:= -\frac{\d ^2}{\d y^2}+\frac{1}{16}y^2
\end{equation}
in $L^2(\mathbb{R})$ (which means the Friedrichs extension of this operator defined on $C^\infty_0(\mathbb{R})$) has purely discrete spectrum: 
\begin{equation}
\label{spectrumHarmonicOscilator}
    \sigma(l) =  \bigg\{D_m:=\frac{1}{2}\left(m-\frac{1}{2}\right) \: | \: m \in \mathbb{N}\bigg\},
\end{equation}
 (cf. for example \cite[Chapter 1, Section 2.3]{HarmonicsOscilator}), and the knowledge of the spectrum of $-\Delta_D^{(0,d)}:$
 \begin{equation}
 \label{spectrum0d}
      \sigma(-\Delta_D^{(0,d)}) =  \bigg\{E_n:=\left(\frac{n \pi}{d}\right)^2 \: | \: n \in \mathbb{N}\bigg\},
 \end{equation}
 we get
\begin{equation*}
    \sigma(L) = \sigma(l) + \sigma(-\Delta_D^{(0,d)}) =  \sigma_{\text{disc}}(l) + \sigma_{\text{disc}}(-\Delta_D^{(0,d)}) = \{D_m + E_n \: | \: n,m \in \mathbb{N}\}.
\end{equation*}
Using the minimax principle we have $\sigma(L) = \sigma_{\text{disc}}(L).$
The discreteness of the spectrum implies that the operator $L$ has compact resolvent and also that the embedding $\iota_1$
\begin{equation*}
    \left(\mathcal{D}(\mathcal{Q}^L), \left(\mathcal{Q}^L[\cdot] + \|\cdot\|_{\mathcal{H}_0}^2\right)^{1/2}\right) \overset{\iota_1}{\hookrightarrow} (\mathcal{H}_0, \|\cdot\|_{\mathcal{H}_0})
\end{equation*}
is compact. As a next step, we show that the embedding $\iota_2$
\begin{equation*}
     \left(\mathcal{D}(J_s^{(0)}), \left(J_s^{(0)}[\cdot] + \|\cdot\|_{\mathcal{H}_0}^2\right)^{1/2}\right) \overset{\iota_2}{\hookrightarrow} \left(\mathcal{D}(\mathcal{Q}^L), \left(\mathcal{Q}^L[\cdot] + \|\cdot\|_{\mathcal{H}_0}^2\right)^{1/2}\right)
\end{equation*}
is bounded. By repeating the same procedure as in the proof of the estimate \eqref{coercivityLions} for $\mathcal{H}_0$ and $\mathcal{H}^1_0$ instead of $\mathcal{H}_1$ and $\mathcal{H}^1_1,$
we get 
\begin{equation*}
     \|\p_y v-\sigma_s\p_z v\|^2_{\mathcal{H}_0} + e^s\|\p_z v\|^2_{\mathcal{H}_0} -E_1e^s\|v\|^2_{\mathcal{H}_0} -\frac{1}{4}\|v\|^2_{\mathcal{H}_0} \geq c_1\|v\|^2_{\mathcal{H}_0^1} - c_2\|v\|^2_{\mathcal{H}_0}
\end{equation*}
for all $v \in \mathcal{H}^1_0$ with $c_1 = \epsilon$ and
\begin{equation*}
    c_2 := \left[E_1e^{s_0}\left(\epsilon + \frac{\epsilon}{1-\epsilon}\|f'\|_{L^\infty(\mathbb{R})}\right)+\frac{1}{4}+\epsilon\right].
\end{equation*}
Therefore, we have 
\begin{align*}
     &\|v\|^2_{\mathcal{H}_0^1} + \frac{1}{16}\|y v\|^2_{\mathcal{H}_0} \leq  C \left[ \|\p_y v-\sigma_s\p_z v\|^2_{\mathcal{H}_0} + e^s\|\p_z v\|^2_{\mathcal{H}_0} -E_1e^s\|v\|^2_{\mathcal{H}_0} +\|v\|^2_{\mathcal{H}_0} + \frac{1}{16}\|y v\|^2_{\mathcal{H}_0}\right],
\end{align*}
where $C = \max \{1/c_1, (c_2-1/4)/c_1, 1\}$ and thus $\iota_2$ is bounded. By composing the bounded embedding $\iota_2$ and the compact embedding $\iota_1$ we obtain $\iota := \iota_2 \circ \iota_1:$
\begin{equation*}
    \left(\mathcal{D}(J_s^{(0)}), \left(J_s^{(0)}[\cdot] + \|\cdot\|_{\mathcal{H}_0}^2\right)^{1/2}\right) \overset{\iota}{\hookrightarrow} (\mathcal{H}_0, \|\cdot\|_{\mathcal{H}_0}),
\end{equation*}
which is therefore also compact. Thus the operator $T_s^{(0)}$ has compact resolvent and also purely discrete spectrum  for every $s \in [0, \infty).$
\end{proof}
\subsection{Asymptotic behaviour of the spectrum}
Next, we need the information about the limit of the eigenvalue $\mu (s)$ as the time $s$ tends to infinity. Using the fact that the function $\sigma_s$ converges in the distributional sense to a multiple of the delta function with support at zero as $s \to \infty,$ as in the paper \cite{TwistedTubesZuazua}, due to the form of  $J^{(0)}_s$ we expect that the corresponding operator $ T_s^{(0)}$ will converge, in a suitable sense, to the one-dimensional operator $l$ defined in \eqref{harmonicOscillator1D} plus an extra Dirichlet boundary condition in zero. More precisely, the limiting operator $l_D$ is defined as the self-adjoint operator in $L^2(\mathbb{R}),$ where the corresponding quadratic form $\mathcal{Q}^l_D$ acts in the same way as the corresponding form of $l,$ however, it has smaller domain
\begin{equation*}
    \mathcal{D}(\mathcal{Q}^l_D) := \{\psi \in \mathcal{D}(\mathcal{Q}^l) \:| \:\psi (0) = 0\}.
\end{equation*}
\begin{remark}
\label{RemarkCoreld}
Alternatively, the form domain $\mathcal{D}(\mathcal{Q}^l_D)$ is the closure of $C_0^\infty(\mathbb{R}\setminus \{0\})$ with respect to the norm $\left(\|\cdot\|^2_{L^2(\mathbb{R})}+\|\nabla\cdot\|^2_{L^2(\mathbb{R})}+\|y\cdot\|^2_{L^2(\mathbb{R})}\right)^{1/2}.$
\end{remark}
Due to the fact that the operators $T^{(0)}_s$ and $l_D$ act in different spaces, we decompose the Hilbert space $\mathcal{H}_0$ into the orthogonal sum
\[\mathcal{H}_0 = \mathfrak{h}_1\oplus \mathfrak{h}_1^{\perp}, \]
where the subspace $\mathfrak{h}_1$ consists of the functions of the form $\psi_1(y,z) = \varphi(y)\mathcal{J}_1(z),$
where $\mathcal{J}_1$ denotes the positive eigenfunction of $-\Delta^{(0,d)}_D$ corresponding to the eigenvalue $E_1$ such that $\|\mathcal{J}_1\|_{L^2(\mathbb{R})} = 1.$
For all $ \psi \in \mathcal{H}_0$ we have the decomposition
\begin{equation}
\label{decomposition}
    \psi = \psi_1 + \phi,
\end{equation}
where $\psi_1 \in \mathfrak{h}_1$ and $\phi \in \mathfrak{h}_1^{\perp}.$
Since the mapping $\pi: \varphi \mapsto \psi_1$ is an isomorphism of $L^2(\mathbb{R})$ onto $\mathfrak{h}_1,$ we may identify any operator $l$ on $L^2(\mathbb{R})$ with the operator $\pi l \pi^{-1}$ on $\mathfrak{h}_1 \subset \mathcal{H}_0.$ 
To show the uniform-resolvent convergence of the operator $T^{(0)}_s$ to $l_D$ we use the following lemma:
\begin{lemma}[\cite{MagneticCazacu}, Lemma A.1.]
\label{lemmaCazacu}
Let $\{T_s\}_{s \in \mathbb{R}}$ be a family of bounded operators on a Hilbert space $\mathcal{H}$ and let $T$ be a compact operator in $\mathcal{H}.$ Suppose that  $\forall \{s_n\}_{n \in \mathbb{N}} \subset \mathbb{R},\: \forall \{f_n\}_{n \in \mathbb{N}}\subset \mathcal{H}$ such that 
\begin{itemize}
    \item $s_n \xrightarrow[n \to \infty]{} \infty,$
    \item $f_n\xrightarrow[n \to \infty]{w}f$ in $\mathcal{H},$
    \item $\forall n \in \mathbb{N}, \:  \|f_n\|_{\mathcal{H}} = 1$
\end{itemize}
implies 
\begin{equation*}
    T_{s_n}f_n \xrightarrow[n\to \infty]{} Tf \text{ in } \mathcal{H}.
\end{equation*}
Then $\{T_s\}_{s \geq 0}$ converges to $T$ uniformly, i.e.
\begin{equation*}
    \lim_{s \to \infty} \|T_s-T\|_{\mathcal{H} \to \mathcal{H}} =0.
\end{equation*}
\end{lemma}
Next proposition enables us to use Lemma \ref{lemmaCazacu}. Due to the fact that we need to use the Hardy inequality $\eqref{RemarkHardy},$ we assume only the non-trivial shear. 
\begin{proposition}
\label{PropositionAsymptotics}
Let $f \in C^{0,1}(\mathbb{R}),$  $f'$ has compact support and $f'\neq 0.$ Then $\forall \{F_s\}_{s\geq 0} \subset \mathcal{H}_0$ such that $F_s \overset{w}{\longrightarrow}F \text{ in } \mathcal{H}_0$ and $\|F_s\|_{H_0} = 1 \: \forall s \geq 0,$ we have 
\[\lim_{s\mapsto \infty}\|(T^{(0)}_s+1)^{-1}F_s-[(l_D+1)^{-1}\oplus 0^{\perp}]F\|_{\mathcal{H}_0} = 0.\]
\end{proposition}
\begin{proof}
The proof is inspired by the proof of Proposition 5.4 in the paper \cite{TwistedTubesZuazua}. For any fixed $F_s \in \mathcal{H}_0$ and sufficiently large positive number $p$ we set $\psi_s := (T^{(0)}_s+p)^{-1}F_s,$ which means, that $\psi_s$ satisfies the resolvent equation
\begin{equation}
\label{asymptoticsresolvent}
    \forall v \in \mathcal{D}(J^{(0)}_s), \quad  J^{(0)}_s(v,\psi_s)+p(v,\psi_s)_{\mathcal{H}_0} = (v,F_s)_{\mathcal{H}_0}.
\end{equation}

If we choose $v:= \psi_s$ we obtain
\begin{align*}
    &\|\p_y\psi_s-\sigma_s\p_z \psi_s\|^2_{\mathcal{H}_0}+\frac{1}{16}\|y\psi_s\|^2_{\mathcal{H}_0}+e^s(\|\p_z \psi_s\|^2_{\mathcal{H}_0}-E_1\|\psi_s\|^2_{\mathcal{H}_0})+p\|\psi_s\|^2_{\mathcal{H}_0} \\
    &=(\psi_s,F_s)_{\mathcal{H}_0}
    \leq \frac{1}{4}\|\psi_s\|^2_{\mathcal{H}_0} + \|F_s\|^2_{\mathcal{H}_0}
    = \frac{1}{4}\|\psi_s\|^2_{\mathcal{H}_0} + 1.
\end{align*}
We rewrite the inequality as
\begin{equation}
\label{boundineq}
    \|\p_y\psi_s-\sigma_s\p_z \psi_s\|^2_{\mathcal{H}_0}+\frac{1}{16}\|y\psi_s\|^2_{\mathcal{H}_0}+e^s(\|\p_z \psi_s\|^2_{\mathcal{H}_0}-E_1\|\psi_s\|^2_{\mathcal{H}_0})+\left(p-\frac{1}{4}\right)\|\psi_s\|^2_{\mathcal{H}_0} \leq 1.
\end{equation}
Henceforth we assume that $p>\frac{1}{4}.$
Next, we use the decomposition 
\[\psi_s(y,z) = \varphi_s(y)\mathcal{J}_1(z) + \phi_s(y,z), \]
where $\phi_s \in \mathfrak{h}^{\perp},$ which implies 
\[\forall y \in \mathbb{R}, \quad (\mathcal{J}_1,\phi_s(y,\cdot))_{L^2((0,d))} =0 .\]
Now for $\epsilon \in (0,1)$ using the orthogonality we obtain
\begin{align*}
\|\p_ z \psi_s\|^2_{\mathcal{H}_0}-E_1\|\psi_s\|^2_{\mathcal{H}_0} 
&= \| \varphi_s\p_ z \mathcal{J}_1\|_{\mathcal{H}_0}+\|\p_ z \phi_s\|^2_{\mathcal{H}_0}-E_1\|\varphi_s \mathcal{J}_1\|_{\mathcal{H}_0}-E_1\|\phi_s\|^2_{\mathcal{H}_0}  \\
&= \epsilon \|\p_ z \phi_s\|^2_{\mathcal{H}_0}+(1-\epsilon)\|\p_ z  \phi_s\|^2_{\mathcal{H}_0}-E_1\|\phi_s\|^2_{\mathcal{H}_0} \\ 
&\geq   \epsilon \|\p_ z \phi_s\|^2_{\mathcal{H}_0}+\left[(1-\epsilon)E_2-E_1\right]\|\phi_s\|^2_{\mathcal{H}_0},
\end{align*}
where $E_2 = \left(\frac{2\pi}{d}\right)^2$ denotes the second eigenvalue of $-\Delta_D^{(0,d)}$ (cf. \eqref{spectrum0d}), and the last inequality follows from the minimax principle. Since $E_1$ is strictly less than $E_2$ we can choose $\epsilon$ so small that $[(1-\epsilon)E_2-E_1]\geq 0$ and use estimate \eqref{boundineq}:
\begin{align*}
    e^s[(1-\epsilon)E_2-E_1]\|\phi_s\|^2_{\mathcal{H}_0}  &\leq \epsilon \|\p_ z \phi_s\|^2_{\mathcal{H}_0}+\left[(1-\epsilon)E_2-E_1\right]\|\phi_s\|^2_{\mathcal{H}_0} \\
    &\leq \|\p_ z \psi_s\|^2_{\mathcal{H}_0}-E_1\|\psi_s\|^2_{\mathcal{H}_0} \\
    &\leq 1
\end{align*}
and thus
\begin{equation}
\label{odhadyphi}
    \|\phi_s\|^2_{\mathcal{H}_0} \leq Ce^{-s},
\end{equation}
where $C$ is a constant depending on $d.$  Similarly we obtain
\begin{align}
\label{odhadydzphi}
  \|\p_ z\phi_s\|^2_{\mathcal{H}_0} &\leq Ce^{-s},\\
\|y\phi_s\|^2_{\mathcal{H}_0} &\leq C, \label{odhadyyphi}\\
\|\varphi_s\|^2_{L^2(\mathbb{R})} &\leq C, \label{odhadyvarphi}\\
\|y\varphi_s\|^2_{L^2(\mathbb{R})} &\leq C. \label{odhadyyvarphi}
\end{align}
As a next step, we define a new function $u_s \in \mathcal{H}_0$ and new variables $(x,z) := (e^{s/2}y,z):$
\[\psi_s(y,z) = e^{s/4}u_s(e^{s/2}y,z).\]
For the form $J^{(0)}_s[\psi_s]$ using the Hardy inequality (cf. Theorem \ref{TheoremHardy} and \eqref{RemarkHardy}) we get
\begin{align*}
  J^{(0)}_s[\psi_s]  &= e^s\|\p_x u_s-f'\p_z u_s\|^2_{\mathcal{H}_0}+\frac{e^{-s}}{16}\|x u_s\|^2_{\mathcal{H}_0}+e^s(\|\p_z u_s\|^2_{\mathcal{H}_0}-E_1\|u_s\|^2_{\mathcal{H}_0})  \\
&\geq e^s\left(\|\p_x u_s-f'\p_z u_s\|^2_{\mathcal{H}_0}+\|\p_z u_s\|^2_{\mathcal{H}_0}-E_1\|u_s\|^2_{\mathcal{H}_0}\right) \\
&\geq e^s c_H \|\rho u_s\|^2_{\mathcal{H}_0}  \\
&= e^s c_H\|\rho_s \psi_s\|^2_{\mathcal{H}_0},
\end{align*}
where $\rho_s(y,z) := \rho (e^{s/2}y,z) $ and $c_H$ is positive. Using the inequality \eqref{boundineq} we obtain
\begin{equation}
\label{rho}
     \|\rho_s \psi_s\|^2_{\mathcal{H}_0} \leq C e^{-s},
\end{equation}
where $C$ depends on $f'$ and $d.$
Furthermore, for $\epsilon \in (0,1)$ we get
\begin{align*}
&\|\p_y\psi_s-\sigma_s\p_z \psi_s\|^2_{\mathcal{H}_0}+e^s(\|\p_ z \psi_s\|^2_{\mathcal{H}_0}-E_1\|\psi_s\|^2_{\mathcal{H}_0}) \\
&\geq \epsilon   \|\p_y\psi_s\|^2_{\mathcal{H}_0}+e^s \int_{\Omega}\left[\left(1-\frac{\epsilon}{1-\epsilon}(f'(e^{s/2}y))^2\right)|\p _z \psi_s(y,z)|^2-E_1|\psi_s|^2\right]\d y \d z,
\end{align*}
 where we used an elementary Young-type inequality in the form $2ab \leq (1-\epsilon)a^2+\frac{b^2}{1-\epsilon}.$
 Using the Fubini's theorem and Poincar\'e-type inequality  on $(0,d)$ we obtain
 \begin{align*}
     &\int_{\Omega}\left(1-\frac{\epsilon}{1-\epsilon}(f'(e^{s/2}y))^2\right)|\p _z \psi_s(y,z)|^2\d y \d z \\
     &\geq  E_1 \left(\int_{\mathbb{R}}\left[1-\frac{\epsilon}{1-\epsilon}(f'(e^{s/2}y))^2\right]\d y \right)\left(\int_0^d| \psi_s(y,z)|^2 \d z\right).
     \end{align*}
For $\epsilon<(1+\|f'\|^2_{L^{\infty}(\mathbb{R})})^{-1}$ the term  $\int_{\Omega}\left(1-\frac{\epsilon}{1-\epsilon}(f'(e^{s/2}y))^2\right)|\p _z \psi_s(y,z)|^2\d y \d z$ is positive and thus we can use previous inequality:
\begin{align}
  &\|\p_y\psi_s-\sigma_s\p_z \psi_s\|^2_{\mathcal{H}_0}+e^s(\|\p_ z \psi_s\|^2_{\mathcal{H}_0}-E_1\|\psi_s\|^2_{\mathcal{H}_0}) \label{odhadrho1} \\
  &\geq \epsilon \|\p_y\psi_s\|^2_{\mathcal{H}_0} - e^s\frac{\epsilon}{1-\epsilon}E_1\|f'\|^2_{L^{\infty}(\mathbb{R})}\|\psi_s\|^2_{L^2(I_s\times (0,d))}, \label{odhadrho2}
\end{align}
where $I_s:= e^{-s/2}I \equiv \{e^{-s/2}x \quad |\quad  x \in I\}$ with $I:= (\inf \text{supp}f',\sup \text{supp}f').$
Using a simple estimate 
\[\|\rho_s\psi_s\|^2_{\mathcal{H}_0} \geq  \min_{I_s}\rho_s\|\psi_s\|^2_{L^2(I_s\times (0,d))}\]
we obtain
\begin{equation}
\label{rhoIs}
 \|\psi_s\|^2_{L^2(I_s\times (0,d))} \leq C\|\rho_s\psi_s\|^2_{\mathcal{H}_0}, 
\end{equation}
where $C$ depends on $I.$
Finally, from 
\eqref{odhadrho1}--\eqref{odhadrho2} and \eqref{boundineq} we get
\begin{equation*}
    \|\p_y\psi_s\|^2_{\mathcal{H}_0} \leq C,
\end{equation*} 
where $C$ is a constant depending on $f'$ and $d.$ Again due to the orthogonality we have separate bounds
\begin{equation}
    \label{odhadyderivace}
    \|\p_y\phi_s\|^2_{\mathcal{H}_0} \leq C, \quad \|\varphi_s'\|^2_{L^2(\mathbb{R})} \leq C.
\end{equation}
From the bound \eqref{odhadyphi} we see that $\phi_s$ converges strongly in $\mathcal{H}_0$ as $s \to 0.$ Moreover, using also the bound \eqref{odhadyderivace} and \eqref{odhadyyphi} we see that $\{\phi_s\}_{s\geq 0}$ is a bounded family in $\mathcal{D}(J^{(0)}_s),$ which implies that there is a subsequence $\phi_{s_{k_n}},$ which converges weakly to zero. Using the strong convergence we observe that   $\phi_s$ converges weakly to zero in $\mathcal{D}(J^{(0)}_s)$ as $s \to \infty.$

Furthermore, the bounds \eqref{odhadyvarphi}--\eqref{odhadyyvarphi} and \eqref{odhadyderivace} imply that $\{\varphi_s\}_{s\geq 0}$ is a bounded family in $\mathcal{D}(\mathcal{Q}^l).$ Therefore, this set is precompact in the weak topology of $\mathcal{D}(\mathcal{Q}^l).$ Next, we denote $\varphi_\infty$ as a weak limit of $\{\varphi_{s_n}\}_{n \in \mathbb{N}}$ in $\mathcal{D}(\mathcal{Q}^l),$  where $s_n$ is an increasing sequence of positive numbers. Since $\mathcal{D}(\mathcal{Q}^l)$ is compactly embedded in $L^2(\mathbb{R})$ (because of the purely discrete spectrum of $l$), we may assume that it converges strongly in $L^2(\mathbb{R})$. From \eqref{rho}, \eqref{odhadrho1}--\eqref{odhadrho2} and the orthogonality we get
\[\|\varphi_s\|^2_{L^2(I_s)}\leq C e^{-s},\]
where $C$ depends on $f'$ and $d.$ As a next step, we multiply the inequality by $e^{s/2}:$ 
\begin{align*}
 e^{s/2}\|\varphi_s\|^2_{L^2(I_s)}  &= \|\varphi_s(e^{-s/2}x)-\varphi_\infty(e^{-s/2}x)\|^2_{L^2(I)}  + \|\varphi_\infty(e^{-s/2}x)\|^2_{L^2(I)} \\
 &+ \left(\varphi_s(e^{-s/2}x)-\varphi_\infty(e^{-s/2}x),\varphi_\infty(e^{-s/2}x)\right)_{L^2(I)}.
\end{align*}
If we take the limit $s \to \infty$ we obtain
\[\varphi_\infty(0) = 0 .\]
Finally, for arbitrary $\varphi \in C_0^\infty(\mathbb{R}\setminus 0)$ (cf. Remark \ref{RemarkCoreld}) we define $v(y,z) :=\varphi(y)\mathcal{J}_1(z) $ as a test function and replace $s$ by $s_n.$ Using the identity \eqref{asymptoticsresolvent}  and sending $n$ to infinity, it is easy to check that
\[(\varphi',\varphi'_\infty)_{L^2(\mathbb{R})} +  \frac{1}{16}(y\varphi,y\varphi_\infty)_{L^2(\mathbb{R})}+p(\varphi,\varphi_\infty)_{L^2(\mathbb{R})} = (\varphi,\hat{F})_{L^2(\mathbb{R})},\]
where $\hat{F}(y) :=(\mathcal{J}_1,F(y,\cdot))_{L^2((0,d))}. $ But it means that $\varphi_\infty = (l_D+p)^{-1}f$ for any weak limit point of $\{\varphi_s\}_{s\geq 0}.$
Consequently, $\psi_s$ converges strongly to $\psi_\infty$ in $\mathcal{H}_0$ as $s \to \infty,$ where $\psi_\infty(y,z) := \varphi_\infty(y)\mathcal{J}_1(z) = [(l_D+p)^{-1}\oplus 0^\perp]F.$
\end{proof}
\begin{theorem}
\label{theoremAsymptotics}
$(T^{(0)}_s+1)^{-1}$ converges to $(l_D+1)^{-1}\oplus 0^{\perp}$ uniformly in $\mathcal{H}_0.$
\end{theorem}
\begin{proof}
Since $l_D$ has purely discrete spectrum, the operator $(l_D+1)^{-1}\oplus 0^{\perp}$ is compact. Finally, we use Proposition \ref{PropositionAsymptotics} and Lemma \ref{lemmaCazacu}.
\end{proof}
\begin{corollary}
\label{Corollary3/4}
  Let $f \in C^{0,1}(\mathbb{R}),$  where $f'$ has compact support and $f' \neq 0.$ Then 
\begin{equation*}
    \lim _{s \to \infty}\mu (s) = 3/4.
\end{equation*}
\end{corollary}
\begin{proof}
Theorem \ref{theoremAsymptotics} implies that $\mu (s)$ converges to the first eigenvalue of $l_D.$ Since the spectrum of the operator $l$ is well known (cf. \eqref{spectrumHarmonicOscilator}), it is easy to see that the first eigenvalue of $l_D$ coincides with the second eigenvalue of $l,$ in particular $3/4.$ (The first eigenvector of $l$ does not belong to the domain $\mathcal{D}(l_D)$ because it does not satisfy the Dirichlet condition at $0,$ however, the second one does).
\end{proof}
\subsection{The improved decay rate}
\label{Sec.ImprovedDecay}
Finally, we can prove Theorem \ref{mainResult}. Corollary \ref{Corollary3/4} implies that for arbitrary small $\epsilon >0$ there is a large $s_\epsilon >0,$ such that $\forall s\geq s_\epsilon : \mu(s) \geq 3/4- \epsilon.$ Thus for fixed $\epsilon >0$ and $\forall s \geq s_\epsilon$ we have:
\begin{align*}
    -\int_0^s\mu(\tau)\d \tau &= -\int_0^{s_\epsilon}\mu(\tau)\d \tau  -\int_{s_\epsilon}^s\mu(\tau)\d \tau  \\
    & \leq -\int_0^{s_\epsilon}\mu(\tau)\d \tau - (3/4-\epsilon)(s-s_\epsilon) \\
    &\leq (3/4-\epsilon)s_\epsilon - (3/4-\epsilon)s,
\end{align*}
where we used in the second estimate the fact that $\mu(s)$ is non-negative, which follows from the positivity of the form $J_s^{(0)}$ for every $s \geq 0.$  At the same time, assuming $\epsilon \leq 3/4$ we obtain for all $s \leq s_\epsilon$
\begin{equation*}
    -\int_0^s\mu(\tau)\d \tau \leq 0 \leq  (3/4-\epsilon)s_\epsilon - (3/4-\epsilon)s. 
\end{equation*}
Using \eqref{energyEstimate} we get $\forall s \in [0,\infty)$
\begin{equation}
\label{energyEstimate3/4}
    \|w(s)\|_{\mathcal{H}_1} \leq C_\epsilon e^{-(3/4-\epsilon)s}\|w_0\|_{\mathcal{H}_1},
\end{equation}
where $C_\epsilon:= e^{s_\epsilon}\geq e^{(3/4-\epsilon)s_\epsilon} $. As a next  step, we return to the original variables $(x,z, t)$ (recall that $w_0 = u_0$). Using the unitarity of the self-similar transformation and the point-wise estimate $K = e^{y^2/4} \geq 1$ we obtain $\forall t \in [0, \infty)$
\begin{equation*}
    \|u(t)\|_{\mathcal{H}_0} = \|w(s)\|_{\mathcal{H}_0} \leq \|w(s)\|_{\mathcal{H}_1} \leq C_\epsilon (1+t)^{-(3/4-\epsilon)s}\|u_0\|_{\mathcal{H}_1}.
\end{equation*}
Finally, because the weight $K = e^{y^2/4}$ depends only on the longitudinal variable, it is invariant by the mapping $\mathcal{L}$ (cf. \eqref{TrafoL}). We conclude with 
\begin{equation*}
    \|S(t)\|_{L^2(\Omega, K) \mapsto L^2(\Omega)} = \sup_{u_0 \in L^2(\Omega, K)\setminus \{0\}}\frac{\|u(t)\|_{L^2(\Omega)}}{\|u_0\|_{L^2(\Omega, K)}} \leq C_\epsilon (1+t)^{-(3/4-\epsilon)},
\end{equation*}
for every $t \in [0, \infty).$ Since we can choose $\epsilon$ arbitrary small we get $\Gamma(\Omega_f) \geq 3/4.$

\subsection{The improved decay rate - an alternative statement}
Main Theorem \ref{mainResult} tells us that the extra polynomial decay rate of solution $u$ of \eqref{heatEquation} in a locally sheared strip is at least three times better that in the straight strip. However, there is no control over the constant in \eqref{mainResultRefolmulated}. In this subsection we present an alternative result, where we get rid of the constant $C_\Gamma$ but we also loose a qualitative knowledge about the decay rate:
\begin{theorem}
\label{alternativeResult}
Let  $f \in C^{0,1}(\mathbb{R}),$ where $f'$ has compact support. Then for every $ t \geq 0$
\begin{equation}
\label{mainAlternativeInequality}
    \|S(t)\|_{L^2(\Omega_f,K) \mapsto L^2(\Omega_f)} \leq (1+t)^{-(\gamma+1/4)},
\end{equation}
where $\gamma \geq 0$ depending on $f'$ and $d.$ Moreover, $\gamma$ is positive if and only if $\Omega_f$ is sheared.
\end{theorem}
In the setting of self-similar solutions (cf. \eqref{selfSimilarTrafo} and \eqref{energyEstimate}) 
we have to show that $\mu (s) \geq 1/4$ for both sheared and straight strip. Thus it is natural to study the shifted operator $T^{(0)}_s-1/4.$ However, it is not obvious from \eqref{formJ0} that such an operator is non-negative. We introduce another unitary transformation $\mathcal{U}_{-1}:\mathcal{H}_0 \to H_{-1},$ which acts in the same way as $\mathcal{U}_0:$
\begin{equation}
\label{unitarytrafoU-1}
    (\mathcal{U}_{-1}v)(y,z):= K^{1/2}(y)v(y,z).
\end{equation}
Next, we introduce the self-adjoint operator $T^{(-1)}_s$ in $\mathcal{H}_{-1}$ via the unitary transformation \eqref{unitarytrafoU-1}
\begin{equation*}
    T^{(-1)}_s:= \mathcal{U}_{-1}T^{(0)}_s({U}_{-1})^{-1}.
\end{equation*}
The operator $T^{(-1)}_s$ is associated with the quadratic form $\mathcal{J}^{(-1)}_s[w]:= \mathcal{J}^{(0)}_s[(\mathcal{U}_{-1})^{-1}w],$ where $w \in \mathcal{D}(\mathcal{J}^{(-1)}_s):=\mathcal{U}_{-1}\mathcal{D}(\mathcal{J}^{(0)}_s).$ Again it is straightforward to check that
\begin{equation}
\label{formJ-1}
    \mathcal{J}^{(-1)}_s[w] = \|\p_y w-\sigma_s\p_z w\|^2_{\mathcal{H}_{-1}}+e^s\|\p_z w\|^2_{\mathcal{H}_{-1}}-E_1e^s\|w\|^2_{\mathcal{H}_{-1}}+\frac{1}{4}\|w\|^2_{\mathcal{H}_{-1}}
\end{equation}
for every $w \in \mathcal{D}(J^{(-1)}_s).$ Now it is easy to see from the structure of the quadratic form that the shifted operator $T^{(-1)}_s-1/4$ is non-negative. Moreover, it is positive if and only if the strip is sheared.
\begin{proposition}
\label{PropositionAlternative}
Let  $f \in C^{0,1}(\mathbb{R}),$ where $f'$ has compact support. If $f' \neq 0,$ then $\mu(s) > 1/4$ for all $s \in [0,\infty).$ On the other hand, in the case $f'=0$ we have $\mu(s) = 1/4$ for all $s \in [0,\infty).$
\end{proposition}
\begin{proof}
Using \eqref{formJ-1} we get $J^{(-1)}_s[w]-\frac{1}{4}\|w\|^2_{\mathcal{H}_{-1}} \geq 0$ for every $w \in \mathcal{D}(J^{(-1)}_s)$ and therefore, using the minimax principle we obtain $\mu(s)\geq 1/4$ for both sheared and straight strip. 

If the strip is not sheared, then $\sigma_s$ is identically zero in $\mathbb{R}$ for all $s \in [0, \infty).$ Choosing $w(y,z) = \mathcal{J}_1(z),$ where $\mathcal{J}_1$ is again the first eigenvector of $-\Delta^{(0,d)}_D$ corresponding to the eigenvalue $E_1,$ we get for straight strip
\begin{equation*}
  \mathcal{J}^{(-1)}_s[\mathcal{J}_1] = +e^s\|\p_z\mathcal{J}_1\|^2_{\mathcal{H}_{-1}}-E_1e^s\|\mathcal{J}_1\|^2_{\mathcal{H}_{-1}}+\frac{1}{4}\|\mathcal{J}_1\|^2_{\mathcal{H}_{-1}}  = \frac{1}{4}\|\mathcal{J}_1\|^2_{\mathcal{H}_{-1}}.
\end{equation*}
Using again the minimax principle we get $\mu(s) \leq 1/4.$ 

It remains to show that $\mu(s) = 1/4$ implies that the strip is straight. Recalling the minimax principle we have
\begin{equation}
\label{minimaxMu=1/4}
    1/4 = \mu(s) = \min_{w \in \mathcal{D}(\mathcal{J}^{(-1)}_s)\setminus\{0\}}\frac{\mathcal{J}^{(-1)}_s[w]}{\|w\|_{\mathcal{H}_{-1}}}.
\end{equation}
Using the Poincar\'e inequality on $(0,d)$ (cf. \eqref{poincare}) and the Fubini's theorem we have 
\[\|\p_z w\|^2_{\mathcal{H}_{-1}} \geq E_1\|w\|^2_{\mathcal{H}_{-1}}\]
for every $w \in \mathcal{D}(\mathcal{J}^{(-1)}_s).$ Therefore, the minimum \eqref{minimaxMu=1/4} is attained by $w \in \mathcal{D}(\mathcal{J}^{(-1)}_s)$ satisfying
\begin{equation}
\label{minimumAttained}
    \|\p_y w-\sigma_s\p_z w\|^2_{\mathcal{H}_{-1}} = 0 \quad \wedge \quad \|\p_zw\|^2_{\mathcal{H}_{-1}}-E_1\|w\|^2_{\mathcal{H}_{-1}} = 0.
\end{equation}
Next, we use the decomposition $w(y,z) = \varphi(y)\mathcal{J}_1(z)+\phi (y,z)$ (cf. \eqref{decomposition}). Therefore, the second equality in \eqref{minimumAttained} implies that $\phi = 0.$ The first equality is then equivalent to 
\begin{align*}
    &\|\dot{\varphi}\|^2_{L^2(\mathbb{R},K^{-1})}\|\mathcal{J}_1\|^2_{L^2((0,d))} +  \|\sigma_s\varphi\|^2_{L^2(\mathbb{R},K^{-1})}\|\p_z\mathcal{J}_1\|^2_{L^2((0,d))}  = 0,
\end{align*}
 and thus $\varphi$ must be a constant and 
\[\|\sigma_s\|^2_{L^2(\mathbb{R},K^{-1})} = 0.\]
The last equality implies that $f$ is a constant and thus the strip is not sheared.
\end{proof}
Now we are able to prove Theorem \ref{mainResultRefolmulated}.
\begin{proof}[Proof of Theorem \ref{mainResultRefolmulated}]
Using Proposition \ref{PropositionAlternative} and Corollary \ref{Corollary3/4} we have
\begin{equation*}
    \gamma := \inf_{s \in [0,\infty)} \mu(s) - 1/4
\end{equation*}
is positive if and only if $\Omega_f$ is sheared. In any case \eqref{energyEstimate} implies
\begin{equation*}
    \|w(s)\|_{\mathcal{H}_1} \leq \|w_0\|_{\mathcal{H}_1}e^{-(\gamma+1/4)s}
\end{equation*}
for every $s \in [0,\infty).$ Using this estimate instead of \eqref{energyEstimate3/4} and using the same procedure as in Subsection \ref{Sec.ImprovedDecay} below \eqref{energyEstimate3/4} we obtain
\begin{equation*}
    \|S(t)\|_{L^2(\Omega, K) \mapsto L^2(\Omega)} \leq  (1+t)^{-(\gamma + 1/4)}
\end{equation*}
for every $t \in [0,\infty).$ This is equivalent to \eqref{mainAlternativeInequality} and $\gamma$ is positive if and only if $\Omega_f$ is sheared. 
\end{proof}
\section{Conclusion}
\label{Conclusion}
In this paper we proved Conjecture \ref{mainConjecture} of \cite{TwistedTubesZuazua} in the case of locally sheared unbounded strips introduced in \cite{ShearBriet}. 
More specifically, we showed that the decay rate of the heat semigroup corresponding to the Dirichlet Laplacian in the unbounded sheared strips is at least three times better that in the case of the straight strip. 
The important ingredient in our proof was
the existence of a geometrically induced Hardy inequality
established in \cite{ShearBriet}.
 
Using the stochastic interpretation of the heat equation, 
our results demonstrate that the expectation lifetime of 
the Brownian particle is made shorter by shearing the strip.
The same conclusion can be made for the effectiveness of 
the temperature cool down of a classical medium enclosed in the sheared strip.
Finally, the heat equation is important for understanding quantum systems as well,
despite the dynamics is intrinsically governed by the Schr\"odinger equation
(cf.~\cite{Simon_1982}).

We conjecture that the inequality of Theorem~\ref{mainResult}
can be replaced by equality (i.e. $\Gamma(\Omega_f)=3/4$ if~$f'\not=0$),
for the decay rates obtained by self-similarity transforms 
are known to be sharp in other circumstances 
(cf.~\cite{Duro-Zuazua_1999,Vazquez-Zuazua_2000}).
An alternative approach to the improved decay rate is given 
by the pointwise estimates for the heat kernel 
performed in~\cite{GrilloKovarik}
and it is expected that the same can be done for the present model too.

Throughout the paper we assumed that $f'$ has compact support. We expect that this hypothesis can be replaced by a vanishing  of $f'$ at infinity 
to get Theorem \ref{mainResult} and Theorem \ref{alternativeResult}. This assumption is known to be enough to ensure the existence of Hardy inequality. However, it is possible that a slower decay of $f'$ at infinity will make the effect of shearing stronger. In particular, is it possible that $\Gamma(\Omega_f)$ is strictly greater that $3/4$ if the strip is sheared and $f'$ tends to zero very slowly at infinity? A different situation will appear if $f'$ does not vanish at infinity. Then the spectrum of the Dirichlet Laplacian can start strictly above $E_1$ (cf. \cite[Theorem 1]{ShearBriet}) and thus there can be even an extra exponential decay rate for the heat semigroup. 
In this case, it is again more natural to study a sub-exponential
decay rate for the semigroup shifted by the lowest point in its spectrum.
Similar spectral-geometric effects have been recently observed
in tubular geometries with globally and asymptotically diverging twisting
\cite{K3,withKolb,K11,BHK,Barseghyan-Khrabustovskyi_2018,KTdA2,KZ3} 
and the study of the associated heat equation constitutes
a challenging open problem.

\subsection*{Acknowledgment}
The research was partially supported 
by the GACR grants No.~18-08835S and 20-17749X.

\bibliographystyle{abbrv}
\bibliography{Sheared}

\end{document}